\documentclass{scrartcl}
\KOMAoptions{paper=letter}

\usepackage[]{group} %

\usepackage{mathtools}
\usepackage{eucal}
\usetikzlibrary{colorbrewer}

\newtheorem*{theorem*}{Theorem} %

\title{The Banks Set and the Bipartisan Set\\ May Be Disjoint}

\author{Felix Brandt \quad Florian Grundbacher\\%
Technical University of Munich, Germany
}

\begin{document}

\maketitle

\begin{abstract}
Tournament solutions play an important role within social choice theory and the mathematical social sciences at large.
We construct a tournament of order $36$ that can be partitioned into the Banks set and the bipartisan set. As a consequence, the Banks set, as well as its refinements, such as the minimal extending set and the tournament equilibrium set, can be disjoint from the bipartisan set.
\end{abstract}

\section{Introduction}

Many problems in the mathematical social sciences can be addressed using tournament solutions, i.e., 
functions that, for each tournament on a finite set of alternatives representing pairwise comparisons, associate
a non-empty subset of the alternatives~\citep[e.g.,][]{Moul86a,Lasl97a,Hudr09a,BBH15a,Suks21a}. Tournament solutions are most prevalent in social choice theory, where pairwise comparisons are typically assumed to be given by the simple majority rule.

Prominent and well-studied tournament solutions include the top cycle, the uncovered set, the Banks set, the bipartisan set, the Copeland set, and the Slater set.
Several studies have addressed whether there are tournaments for which two given tournament solutions return disjoint choice sets. 
For example, the first tournament proposed in the literature for which the Banks set and the Slater set are disjoint is of order~$75$ \citep{LaLa91a}.\footnote{\citet{LaLa91a} presented a similar tournament on $139$ alternatives in which the Banks set, the Slater set, and the Copeland set are all disjoint from each other.} Later, the order
was improved to $16$ by \citet{CHW97a} and then to~$14$ by \citet{OsVa10a}. \citeauthor{OsVa10a} have also provided a lower bound of~$11$ on the order by means of an exhaustive computer analysis.
For another example, \citet{Moul86a} presented a tournament of order $13$ in which the Banks set and the Copeland set %
are disjoint and \citet{Hudr99a} proved that this tournament is minimal.\footnote{The tournament given in the present paper strengthens this statement because the Copeland scores of all alternatives in the Banks set of this tournament are below average.}
\citet{BDS13a} addressed these questions systematically by exhaustive computer search and defined the \emph{separation index} of two tournament solutions as the order of the smallest tournament for which the choice sets returned by the two tournament solutions are disjoint.
Whether the Banks set and the bipartisan set always intersect has been identified as a challenging problem by \citet{LLL95a}, who answered the same question for many pairs of tournament solutions.\footnote{For example, they show that the minimal covering set (a coarsening of the bipartisan set) always intersects with the Banks set.}
They also note the weaker open problem concerning the relationship between the tournament equilibrium set and the bipartisan set.
\citet[][p.~42]{BDS13a} write that ``perhaps the most interesting open problem regarding the relationships between tournament solutions concerns the bipartisan set and the Banks set.''
They provide the first tournament in which the Banks set and the bipartisan set are not contained in each other and prove that their separation index is at least 11.

In this paper, we show that the separation index is at most $36$ by providing the first tournament in which the Banks set and the bipartisan set are disjoint. It follows from inclusions proven by \citet{Schw90a} and \citet{BHS15a}, respectively, that both the tournament equilibrium set and the minimal extending set can also be disjoint from the bipartisan set. 
In his monograph on tournament solutions, \citet[][p.~235]{Lasl97a} concludes his summary of the 
set-theoretic relationships between tournaments solutions by asserting that ``it is not known whether the intersection of [the Banks set] and [the bipartisan set] can be empty.''
Similarly, \citet[][p.~76]{BBH15a} laments that ``the exact location of [the bipartisan set] in this diagram [their Fig.~3.7] is unknown.''
Thanks to the tournament described in this paper and the stated inclusions, this is no longer the case.

\section{Preliminaries}

A \emph{tournament $T$} is a pair $(A,{\succ})$, where~$A$ is a finite set of alternatives and~$\succ$ a binary relation on~$A$ %
that is both \emph{asymmetric} ($x\succ y$ implies not $y\succ x$ for all $x,y\in A$) and \emph{connex} ($x\neq y$ implies $x\succ y$ or $y\succ x$ for all $x,y\in A$).
The relation $\succ$ is usually referred to as \emph{dominance relation} and we say that $x$ dominates $y$ when $x\succ y$.
We write $x\succ B$ if $x\succ y$ for all $y\in B$
and refer to the largest set $B\subseteq A$ for which $x\succ B$ as the \emph{dominion} of $x$.
The \emph{order} of a tournament $T=(A,{\succ})$ refers to the cardinality of $A$.

For a given tournament $T=(A,{\succ})$, a subset of alternatives $B\subseteq A$ is called \emph{transitive} if %
$x\succ y$ and $y\succ z$ imply $x\succ z$ for all $x,y,z\in B$. 
The \emph{Banks set} $\ba(T)$ of a tournament, named after \citet{Bank85a}, is then defined as the set of all alternatives that are maximal in some inclusion-maximal transitive subset, i.e.,
\[
\ba(T) = \{ x\in A \colon \exists B\subseteq A,\; B \text{ is transitive},\; x\succ B\setminus\{x\} \text{, and }\nexists y\in A \text{ with } y\succ B\}\text{.}
\]

Let $M_T\in \{-1,0,1\}^{A\times A}$ %
be the \emph{skew-adjacency matrix} of $T$, where 
\[
M_T(x,y) = 
\begin{cases}
1 & \text{if }x\succ y\text,\\
-1 & \text{if }y\succ x\text{, and}\\
0 & \text{if }x=y\text.
\end{cases}
\]
\citet{LLL93b} and \citet{FiRy95a} have shown independently that every tournament $T$ admits a unique probability distribution $p_T\in [0,1]^A$ such that $\sum_{x\in A} p_T(x)=1$ and
$\sum_{x\in A} p_T(x)M_T(x,y)\ge 0$ for all $y\in A$.
The distribution $p_T$ %
corresponds to the unique mixed Nash equilibrium of the symmetric zero-sum game $M_T$ (whose value is $0$). The 
\emph{bipartisan set} $\bp(T)$ of a tournament $T$ is defined as the support of this equilibrium, i.e., 
\[ \bp(T) = \{ x\in A \colon p_T(x)>0\}\text{.}
\]

\section{The Result}

We start by describing some observations that led to the construction of our example tournament, which may be skipped by the impatient reader.
It easily follows from known facts about tournament solutions that if there is a tournament in which the Banks set and the bipartisan set are disjoint, then there also exists a tournament $T$ in which $\ba(T)$ and $\bp(T)$ partition the tournament and in which the probability distribution $p_T$ is uniform \citep[see, e.g.,][Prop.~6.2.4, Thm. 6.3.2, and Thm. 7.1.3]{Lasl97a}.
Moreover, the definition of $p_T$ and its uniformity on $\bp(T)$ imply that $T$ restricted to the alternatives in $\bp(T)$ is regular, and that any alternative outside of $\bp(T)$ is dominated by a majority of the alternatives in $\bp(T)$. Next, to ensure that no alternative from $\bp(T)$ lies in $\ba(T)$, we need every (maximal) transitive subset of $\bp(T)$ to be dominated by some alternative in $\ba(T)$ (as otherwise the maximal alternative in such a subset would belong to $\ba(T)$). In particular, $\bp(T)$ cannot contain a transitive subset with more than half the alternatives from $\bp(T)$. In our example tournament, $\bp(T)$ consists of $9$ alternatives with $27$ maximal transitive subsets, each containing $4$ alternatives. Therefore, we need a unique alternative in $\ba(T)$ for each of these subsets that dominates all alternatives in the respective subset (hence the $27$ alternatives in $\ba(T)$). The tricky part is to ensure that these transitive subsets are still dominated even if we add further non-maximal alternatives from $\ba(T)$.
Our attempts to construct a smaller example using this approach were futile. 

\begin{theorem*}\label{thm:t}
There exists a tournament $T=(A,{\succ})$ of order $36$ such that $\ba(T)\cap \bp(T)=\emptyset$ and $\ba(T)\cup\bp(T)=A$.
\end{theorem*}

\begin{proof}
We begin with the formal description of the tournament $T=(A,{\succ})$. We define $A \coloneqq \{ v^i_{j,k} \colon i \in \{0,1,2,3\}, j,k \in \{1,2,3\} \}$ and let $\succ$ be given as follows. To simplify notation, we set $3 + 1 \coloneqq 1$ and $1 - 1 \coloneqq 3$ and write $*$ when universally quantifying over all possible indices whenever we consider alternatives in $A$.
For all $i \in \{0,1,2,3\}$ and $j,k,\ell \in \{1,2,3\}$, we have
\begin{alignat*}{4}
v^i_{j,k} &\succ v^i_{j,k+1} ~ && \text,\\
v^i_{j,*} &\succ v^i_{j+1,*} && \text,\\
v^0_{j,*} &\succ v^j_{*,*} && \text,\\
v^0_{j,k} &\succ v^{j-1}_{\ell,*} && \text{if and only if } k=\ell\text,\\
v^0_{j,k} &\succ v^{j+1}_{*,\ell} && \text{if and only if } k=\ell\text{, and}\\
v^j_{k,*} &\succ v^{j-1}_{*,\ell} && \text{if and only if } k=\ell-1\text.
\end{alignat*}

The following subsets of alternatives will be of particular interest. For all $i \in \{0,1,2,3\}$ and $j \in \{1,2,3\}$, we write
\[
\Delta^i \coloneqq \{v^i_{k,\ell} : k,\ell \in \{1,2,3\}\} \quad\text{and}\quad \Delta^i_j \coloneqq \{v^i_{j,k}\colon k \in \{1,2,3\}\}\text.
\]
The structure of $T$ is visualized in \Cref{fig:t}. 

\begin{figure}[tbp]
\centering
\begin{tikzpicture}[>=Latex]
\def\nodewidth{12pt}
\def\compwidth{50pt]}
\def\bigcompwidth{140pt}
\tikzstyle{node}=[minimum width=\nodewidth,inner sep=0,circle,draw,font={\scriptsize\sffamily}]
\tikzstyle{bigcomp}=[minimum width=\bigcompwidth,inner sep=0,circle,fill=black!10,font={\sffamily}]
\tikzstyle{comp}=[minimum width=\compwidth,inner sep=0,circle,fill=black!10,font={\footnotesize\sffamily}]
\def\a{5.1};
\def\b{1.3};
\def\c{0.5};
\def\col{Set1-D!75};
\def\cola{Set1-D!15};
\def\coldx{Set1-B}; %
\def\coldy{Set1-C}; %
\def\coldz{Set1-A}; %

\node (0) at (0,0) {0};
\node[bigcomp] (1) at (30:\a) {1};
\node[bigcomp] (2) at (150:\a) {2};
\draw (-90:\a) node[bigcomp, fill=\cola] (3) {3};

\draw (0)+(-90:\b) node[comp,fill=\cola] (01) {1}
      (0)+(30:\b) node[comp] (02) {2}
      (0)+(150:\b) node[comp] (03) {3};
\draw[->] (01) edge (02) (02) edge (03) (03) edge (01);

\draw (01)+(30:\c) node[node] (011) {1}
      (01)+(150:\c) node[node] (012) {2}
      (01)+(-90:\c) node[node] (013) {3};
\draw[->] (011) edge (012) (012) edge (013) (013) edge (011);
      \draw[ultra thick,draw=\coldz] (011)+(90:0.5*\nodewidth) arc [radius=0.5*\nodewidth, start angle=90, end angle=270];
      \draw[ultra thick,draw=\coldy] (011)+(270:0.5*\nodewidth) arc [radius=0.5*\nodewidth, start angle=-90, end angle=90];
      \draw[ultra thick,draw=\coldx] (012)+(90:0.5*\nodewidth) arc [radius=0.5*\nodewidth, start angle=90, end angle=270];
      \draw[ultra thick,draw=\coldz] (012)+(-90:0.5*\nodewidth) arc [radius=0.5*\nodewidth, start angle=-90, end angle=90];
      \draw[ultra thick,draw=\coldy] (013)+(-90:0.5*\nodewidth) arc [radius=0.5*\nodewidth, start angle=-90, end angle=90];
      \draw[ultra thick,draw=\coldx] (013)+(-270:0.5*\nodewidth) arc [radius=0.5*\nodewidth, start angle=-270, end angle=-90];

\draw (02)+(150:\c) node[node] (021) {1}
      (02)+(-90:\c) node[node] (022) {2}
      (02)+(30:\c) node[node] (023) {3};
\draw[->] (021) edge (022) (022) edge (023) (023) edge (021);

\draw (03)+(-90:\c) node[node,fill=\cola, draw=\coldx,thick] (031) {1}
      (03)+(30:\c) node[node] (032) {2}
      (03)+(150:\c) node[node,fill=\col,draw=none] (033) {\textcolor{white}3}; 
      \draw[ultra thick,draw=\coldz] (033)+(30:0.5*\nodewidth) arc [radius=0.5*\nodewidth, start angle=30, end angle=150];
      \draw[ultra thick,draw=\coldx] (033)+(150:0.5*\nodewidth) arc [radius=0.5*\nodewidth, start angle=150, end angle=270];
      \draw[ultra thick,draw=\coldy] (033)+(-90:0.5*\nodewidth) arc [radius=0.5*\nodewidth, start angle=-90, end angle=30];
      \draw[ultra thick,draw=\coldz] (031)+(30:0.5*\nodewidth) arc [radius=0.5*\nodewidth, start angle=30, end angle=150];
      \draw[ultra thick,draw=\coldx] (031)+(150:0.5*\nodewidth) arc [radius=0.5*\nodewidth, start angle=150, end angle=270];
      \draw[ultra thick,draw=\coldy] (031)+(-90:0.5*\nodewidth) arc [radius=0.5*\nodewidth, start angle=-90, end angle=30];
\draw[->] (031) edge (032) (032) edge (033) (033) edge (031);

\draw (1)+(90:\b) node[comp,draw=gray] (11) {1}
      (1)+(-30:\b) node[comp,draw=gray] (12) {2}
      (1)+(-150:\b) node[comp,draw=gray] (13) {3};
\draw[->] (11) edge (12) (12) edge (13) (13) edge (11);

\draw (11)+(-150:\c) node[node] (111) {1}
      (11)+(-30:\c) node[node] (112) {2}
      (11)+(90:\c) node[node,fill=\cola] (113) {3};
      \draw[ultra thick,draw=\coldz] (113)+(30:0.5*\nodewidth) arc [radius=0.5*\nodewidth, start angle=30, end angle=150];
      \draw[ultra thick,draw=\coldx] (113)+(150:0.5*\nodewidth) arc [radius=0.5*\nodewidth, start angle=150, end angle=270];
      \draw[ultra thick,draw=\coldy] (113)+(-90:0.5*\nodewidth) arc [radius=0.5*\nodewidth, start angle=-90, end angle=30];

\draw (12)+(90:\c) node[node] (121) {1}
      (12)+(-150:\c) node[node] (122) {2}
      (12)+(-30:\c) node[node,fill=\cola] (123) {3};
      \draw[ultra thick,draw=\coldz] (123)+(30:0.5*\nodewidth) arc [radius=0.5*\nodewidth, start angle=30, end angle=150];
      \draw[ultra thick,draw=\coldx] (123)+(150:0.5*\nodewidth) arc [radius=0.5*\nodewidth, start angle=150, end angle=270];
      \draw[ultra thick,draw=\coldy] (123)+(-90:0.5*\nodewidth) arc [radius=0.5*\nodewidth, start angle=-90, end angle=30];

\draw (13)+(-30:\c) node[node] (131) {1}
      (13)+(90:\c) node[node] (132) {2}
      (13)+(-150:\c) node[node, fill=\cola] (133) {3};
      \draw[ultra thick,draw=\coldz] (133)+(30:0.5*\nodewidth) arc [radius=0.5*\nodewidth, start angle=30, end angle=150];
      \draw[ultra thick,draw=\coldx] (133)+(150:0.5*\nodewidth) arc [radius=0.5*\nodewidth, start angle=150, end angle=270];
      \draw[ultra thick,draw=\coldy] (133)+(-90:0.5*\nodewidth) arc [radius=0.5*\nodewidth, start angle=-90, end angle=30];

\draw (2)+(-150:\b) node[comp,draw=gray] (21) {1}
      (2)+(90:\b) node[comp,draw=gray] (22) {2}
      (2)+(-30:\b) node[comp, fill=\cola] (23) {3};
\draw[->] (21) edge (22) (22) edge (23) (23) edge (21);
      \draw[ultra thick,draw=\coldz] (23)+(30:0.5*\compwidth) arc [radius=0.5*\compwidth, start angle=30, end angle=150];
      \draw[ultra thick,draw=\coldx] (23)+(150:0.5*\compwidth) arc [radius=0.5*\compwidth, start angle=150, end angle=270];
      \draw[ultra thick,draw=\coldy] (23)+(-90:0.5*\compwidth) arc [radius=0.5*\compwidth, start angle=-90, end angle=30];

\draw (21)+(-30:\c) node[node] (211) {1}
      (21)+(90:\c) node[node] (212) {2}
      (21)+(-150:\c) node[node] (213) {3};

\draw (22)+(-150:\c) node[node,fill=\coldx!50] (221) {1} %
      (22)+(-30:\c) node[node,fill=\coldy!50] (222) {2} %
      (22)+(90:\c) node[node,fill=\coldz!50] (223) {3}; %

\draw (23)+(90:\c) node[node] (231) {1}
      (23)+(-150:\c) node[node] (232) {2}
      (23)+(-30:\c) node[node] (233) {3};

\draw (3)+(-30:\b) node[comp,fill=\cola] (31) {1}
      (3)+(-150:\b) node[comp,fill=\cola] (32) {2}
      (3)+(90:\b) node[comp,fill=\cola] (33) {3};
\draw[->] (31) edge (32) (32) edge (33) (33) edge (31);      
      
      \draw[ultra thick,draw=\coldx] (31)+(90:0.5*\compwidth) arc [radius=0.5*\compwidth, start angle=90, end angle=270];
      \draw[ultra thick,draw=\coldz] (31)+(-90:0.5*\compwidth) arc [radius=0.5*\compwidth, start angle=-90, end angle=90];
      \draw[ultra thick,draw=\coldy] (32)+(-90:0.5*\compwidth) arc [radius=0.5*\compwidth, start angle=-90, end angle=90];
      \draw[ultra thick,draw=\coldx] (32)+(-270:0.5*\compwidth) arc [radius=0.5*\compwidth, start angle=-270, end angle=-90];
      \draw[ultra thick,draw=\coldz] (33)+(90:0.5*\compwidth) arc [radius=0.5*\compwidth, start angle=90, end angle=270];
      \draw[ultra thick,draw=\coldy] (33)+(270:0.5*\compwidth) arc [radius=0.5*\compwidth, start angle=-90, end angle=90];

\draw (31)+(90:\c) node[node] (311) {1}
      (31)+(-150:\c) node[node] (312) {2}
      (31)+(-30:\c) node[node] (313) {3};

\draw (32)+(-30:\c) node[node] (321) {1}
      (32)+(90:\c) node[node] (322) {2}
      (32)+(-150:\c) node[node] (323) {3};

\draw (33)+(-150:\c) node[node] (331) {1}
      (33)+(-30:\c) node[node] (332) {2}
      (33)+(90:\c) node[node] (333) {3};

\draw[->,out=0,in=240] (01) to (1);
\draw[->,out=120,in=0] (02) to (2);
\draw[->,out=240,in=120] (03) to (3);

\draw[->] (012) edge[bend right=10] (32) (013) edge (33) (011) edge[bend left=10] (31);
\draw[->] (022) edge[bend right=10] (12) (023) edge (13) (021) edge[bend left=10] (11);
\draw[->] (032) edge[bend right=10] (22) (033) edge (23) (031) edge[bend left=10] (21);

\end{tikzpicture}

\caption{Incomplete depiction of tournament $T=(A,\succ)$ with $\ba(T)\cap \bp(T)=\emptyset$. 
The 36 alternatives $\{v^i_{j,k} \colon i \in \{0,1,2,3\}, j,k \in \{1,2,3\}\}$ are labeled such that $i$ denotes the outmost triangle $\Delta^i$, $j$ the triangle $\Delta^i_j$ within $i$, and $k$ the alternative within triangle $\Delta^i_j$.
The automorphism group of the tournament admits only two orbits: $\Delta^0=\{v^0_{j,k} \colon j,k \in \{1,2,3\}\}$ and $A\setminus \Delta^0$.
$\bp(T)=\Delta^0$ and $\ba(T)=A\setminus \Delta^0$.
The dominion of the purple alternative $v^0_{3,3}$ is highlighted in light purple. Every transitive subtournament of this dominion is dominated by one of the three alternatives in $\Delta^2_2$, highlighted in blue, green, and red. The dominions of these alternatives within the dominion of $v^0_{3,3}$ are indicated by respectively colored boundaries.
}
\label{fig:t}

\end{figure}

In order to identify the elements of $\bp(T)$, consider the probability distribution that assigns probability $\nicefrac19$ to the alternatives in $\Delta^0$ and probability 0 to all other alternatives. We claim that this distribution satisfies $\sum_{x\in A} p_T(x)M_T(x,y)\ge 0$ for all $y \in A$. All $y\in \Delta^0$ dominate four alternatives in $\Delta^0$ and are dominated by four alternatives in $\Delta^0$, respectively. Hence, $\sum_{x\in A} p_T(x)M_T(x,y)=0$. All $y \in A \setminus \Delta^0$ dominate precisely four alternatives in $\Delta^0$ and are dominated by five alternatives in $\Delta^0$. For example, $v^1_{1,3}$ is dominated by the three alternatives in $\Delta^0_1$, by $v^0_{2,1}$, and by $v^0_{3,3}$. Hence, $\sum_{x\in A} p_T(x)M_T(x,y)=\nicefrac 19> 0$.
Per definition, $\bp(T)$ consists of all alternatives $x \in A$ with $p_T(x)>0$, giving $\bp(T)=\Delta^0$.

To simplify the computation of $\ba(T)$, we first observe that $T$ admits the following four %
graph automorphisms: %
the map $\varphi \colon A \to A$, defined for all $i,j,k \in \{1,2,3\}$ by
\[ 
    \varphi(v^0_{i,j}) \coloneqq v^0_{i+1,j}
    \quad\text{and}\quad
    \varphi(v^i_{j,k}) \coloneqq v^{i+1}_{j,k},
\]
and for any $\ell \in \{1,2,3\}$ the map $\psi_\ell \colon A \to A$,
defined for all $i,j,k \in \{1,2,3\}$ by
\[
    \psi_\ell(v^0_{i,j})\coloneqq
        \begin{cases}
            v^0_{i,j+1} & \text{if } i = \ell\text, \\
            v^0_{i,j}   & \text{if } i \neq \ell\text,
        \end{cases}
    \quad\text{and}\quad
    \psi_\ell(v^i_{j,k})\coloneqq
        \begin{cases}
            v^i_{j,k}   & \text{if } i= \ell\text, \\
            v^i_{j+1,k} & \text{if } i = \ell-1\text{, and}\\
            v^i_{j,k+1} & \text{if } i = \ell+1\text.
        \end{cases}
\]
Intuitively, $\varphi$ rotates the entire tournament counterclockwise by $120^\circ$. $\psi_\ell$ rotates the alternatives in the small triangles $\Delta^0_\ell$ and $\Delta^{\ell+1}_*$ counterclockwise,
and the entire small triangles in $\Delta^{\ell-1}$ clockwise. We skip the details of checking that these four maps are indeed graph automorphisms. 

It is straightforward to see that for any pair $x,y \in \Delta^0$, an appropriate concatenation of these four automorphisms maps $x$ to $y$. In fact, the automorphism group induced by these maps admits only two equivalence classes (orbits): $\Delta^0$ and $A\setminus\Delta^0$. In order to prove that $\ba(T)\cap \bp(T)=\emptyset$, it thus suffices to show that an arbitrary element of $\bp(T)=\Delta^0$ fails to be contained in $\ba(T)$. 
To this end, consider alternative $v^0_{3,3}$, highlighted in dark purple in \Cref{fig:t}. The dominion of this alternative is highlighted in light purple and consists of $v^0_{3,1}$, $v^1_{1,3}$, $v^1_{2,3}$, $v^1_{3,3}$,
$\Delta^2_3$, $\Delta^0_1$, and $\Delta^3$.
We claim that every transitive subset $B$ of this dominion is dominated by one of the three alternatives in $\Delta^2_2$, highlighted in blue, green, and red in \Cref{fig:t}. The dominions of these alternatives within the dominion of $v^0_{3,3}$ are indicated by respectively colored boundaries. First, observe that each of the alternatives in $\Delta^2_2$ dominates $\Delta^2_3$, $v^0_{3,1}$, $v^0_{3,3}$, $v^1_{1,3}$, $v^1_{2,3}$, and $v^1_{3,3}$. %
We can thus focus on the alternatives in $\Delta^0_1$ and $\Delta^3$.
Since both of these sets represent triangles (of subtriangles), $B$ can only contain at most two elements of $\Delta^0_1$ and elements of at most two of the subtriangles $\Delta^3_1$, $\Delta^3_2$, and $\Delta^3_3$ of $\Delta^3$. %
Additionally note %
that all edges between $\Delta^0_1$ and $\Delta^3$ that are not depicted in \Cref{fig:t} go from $\Delta^3$ to $\Delta^0_1$.
Thus, for any $j \in \{1,2,3\}$, the triples $v^0_{1,j}$, $v^0_{1,j+1}$, $\Delta^3_{j+1}$ and $v^0_{1,j}$, $\Delta^3_j$, $\Delta^3_{j+1}$ represent further triangles. From this, it is straightforward to verify that all %
transitive subsets of $\Delta^0_1\cup \Delta^3$ are subsets of %
$\{v^0_{1,2},v^0_{1,3}\}\cup\Delta^3_1\cup\Delta^3_2$,
or %
$\{v^0_{1,1},v^0_{1,3}\}\cup\Delta^3_2\cup\Delta^3_3$,
or %
$\{v^0_{1,1},v^0_{1,2}\}\cup\Delta^3_1\cup\Delta^3_3$. 
The first set is dominated by $v^2_{2,1}$, the second by $v^2_{2,2}$, and the third by $v^2_{2,3}$, as can be seen by the boundary colors in \Cref{fig:t}.
We conclude that $v^0_{3,3}\not\in \ba(T)$ %
and, by the automorphism argument from above, $\Delta^0 \cap \ba(T) = \emptyset$. %

Since $\ba(T)\neq \emptyset$, some alternative in $A\setminus \Delta^0$ has to be contained in $\ba(T)$. Moreover, since for any pair $x,y \in A\setminus \Delta^0$, an appropriate concatenation of the four automorphisms described above maps $x$ to $y$, we obtain that $\ba(T)=A\setminus \Delta^0$.
\end{proof}

It can be shown that the theorem holds independently of how the edges in each of the small outer triangles $\Delta^*_*$ are oriented. In particular, they need not form cycles. We chose cycles to simplify the proof by exploiting automorphisms. 
Alternatively, the argument for $v^0_{3,3}$ can be generalized to any alternative in $\Delta^0$. This is possible because the edges in each of the small outer triangles $\Delta^*_*$ did not play a role in the proof for $v^0_{3,3} \notin \ba(T)$.

While the tournament $T$ was found manually, we have additionally verified the correctness of the statement using a computer.
Remarkably, the degree (or Copeland score) of each alternative in $\bp(T)$ is 19, while that of each alternative in $\ba(T)$ is 17. Hence, in this tournament, the alternatives in $\ba(T)$ are precisely those with below-average degrees.

The two remaining open questions concerning the inclusion relationships between commonly studied tournament solutions are whether the tournament equilibrium set is contained in the minimal extending set and whether the minimal extending set is contained in the minimal covering set \citep[see][]{BHS15a}. Another open problem is whether the Banks set satisfies Kelly-strategyproofness \citep[see][]{BBH15a}.

\subsection*{Acknowledgments}
This material is based on work supported by the Deutsche Forschungsgemeinschaft under grants {BR~2312/11-2} and {BR~2312/12-1}.


\end{document}